\documentclass[10pt,twocolumn,twoside]{IEEEtran}
\usepackage{amsmath,amssymb,euscript,yfonts,psfrag,latexsym,dsfont,graphicx,bbm,color,amstext,wasysym,pdfsync} 
\usepackage{epstopdf}

\begin{document}
\newtheorem{thm}{Theorem}
\newtheorem{conj}[thm]{Conjecture}
\newtheorem{cor}[thm]{Corollary}
\newtheorem{lemma}[thm]{Lemma}
\newtheorem{prop}[thm]{Proposition}
\newtheorem{problem}[thm]{Problem}
\newtheorem{remark}[thm]{Remark}
\newtheorem{defn}[thm]{Definition}
\newtheorem{ex}[thm]{Example}
\newcommand{\ignore}[1]{}
\newcommand{\mR}{{\mathbb R}}
\newcommand{\mN}{{\mathbb N}}
\newcommand{\mI}{{\mathbb I}}
\newcommand{\supp}{{\rm supp}}
\newcommand{\ext}{{\rm ext}}
\newcommand{\bR}{{\bold R}}
\newcommand{\rR}{{\bold R}}
\newcommand{\rG}{{\bold G}}
\newcommand{\mZ}{{\mathbb Z}}
\newcommand{\mC}{{\mathbb C}}
\newcommand{\mT}{{\mathbb T}}
\newcommand{\fR}{{\mathfrak R}}
\newcommand{\fK}{{\mathfrak K}}
\newcommand{\fG}{{\mathfrak G}}
\newcommand{\fT}{{\mathfrak T}}
\newcommand{\fC}{{\mathfrak C}}
\newcommand{\fF}{{\mathfrak F}}
\newcommand{\cF}{{\mathcal F}}
\newcommand{\fN}{{\mathfrak N}}
\newcommand{\fM}{{\mathfrak M}}
\newcommand{\mD}{{\mathbb{D}}}
\newcommand{\cU}{{\mathcal{U}}}
\newcommand{\cL}{{\mathcal{L}}}
\newcommand{\cE}{{\mathcal{E}}}
\newcommand{\cI}{{\mathcal{I}}}
\newcommand{\cJ}{{\mathcal{J}}}
\newcommand{\cD}{{\mathcal{D}}}
\newcommand{\cR}{{\mathcal{R}}}
\newcommand{\cG}{{\mathcal{G}}}
\newcommand{\bJ}{{\mathbb{J}}}
\newcommand{\bI}{{\mathbb{I}}}
\newcommand{\pC}{C_{\rm perio}((-\pi,\pi])}
\newcommand{\s}{{\rm s}}
\newcommand{\dtheta}{{\frac{d\theta}{2\pi}}}
\newcommand{\ar}{{\rm a}}
\newcommand{\intpi}{{\int_{-\pi}^\pi}}
\newcommand{\me}{{\rm _{ME}}}
\newcommand{\mr}{{\rm _{sm}}}
\def\spacingset#1{\def\baselinestretch{#1}\small\normalsize}
\newcommand{\zn}{{0:n}}
\newcommand{\bczn}{{{\bf c}_{\zn}}}
\newcommand{\bc}{{{\bf c}}}

\newcommand{\zk}{{0:k}}
\newcommand{\bczk}{{{\bf c}_{\zk}}}
\newcommand{\zkEtt}{{0:k-1}}
\newcommand{\bczkEtt}{{{\bf c}_{\zkEtt}}}

\title{Uncertainty Bounds for Spectral Estimation}
\author{Johan Karlsson and Tryphon T. Georgiou \thanks{Supported by
Swedish Research Council, G\"oran
Gustafsson Foundation, National Science Foundation  Grant No. CCF-1218388 and Grant No. 1027696, Air Force
Office of Scientific Research under Grant FA9550-12-1-0319, and the Vincentine Hermes-Luh Endowment.}  \thanks{Johan Karlsson is with the department of Electrical and Computer Engineering, University of Florida, Gainesville, Florida 32611, {\tt
jkarlsson@ufl.edu}. Tryphon T.  Georgiou is with the
Department of Electrical Engineering, University of Minnesota,
Minneapolis, Minnesota 55455, USA, {\tt tryphon@ece.umn.edu}}} \date{}
\maketitle

\spacingset{.96}

\begin{abstract} 
The purpose of this paper is to study metrics suitable for assessing uncertainty of power spectra when these are based on finite second-order statistics. The family of power spectra which is consistent with a given range of values for the estimated statistics represents the {\em uncertainty set} about the ``true'' power spectrum. Our aim is to quantify the size of this uncertainty set using suitable notions of distance, and in particular, to compute the diameter of the set since this represents an upper bound on the distance between any choice of a nominal element in the set and the ``true'' power spectrum. Since the uncertainty set may contain power spectra with lines and discontinuities, it is natural to quantify distances in the weak topology---the topology defined by continuity of moments. We provide examples of such weakly-continuous metrics and focus on particular metrics for which we can explicitly quantify spectral uncertainty. We then consider certain high resolution techniques which utilize filter-banks for pre-processing, and compute  worst-case {\em a priori} uncertainty bounds solely on the basis of the filter dynamics. This allows the {\em a priori} tuning of the filter-banks for improved resolution over selected frequency bands.
\end{abstract}


\section{Introduction}\label{sec:introduction}

\PARstart{I}{n} practice, the estimation of power spectra in stationary time-series often relies on
second-order statistics. The premise is that these are moments of an
underlying power spectral distribution ---the true power spectrum. Thus,
the question arises as to how much is ``knowable'' about the distribution of power in the spectrum from such statistics. 

Asymptotically, as more data accrue the convergence is guaranteed in a suitable sense, but the practical question remains on how to bound the error when only limited information is available. To this end, it is important to consider how a finite set of statistics localizes the power spectrum.
Traditionally, for many applications, one relies on a particular power spectrum selected out of a variety of methods
that lead to specific choices, all consistent (in different ways) with the recorded data and the estimated
moments. Historically, Burg's algorithm and the maximum entropy spectrum, and the Pisarenko harmonic decomposition are specific such choices
\cite{Haykin,StoicaMoses} and so are the correlogram and the periodogram.
Thus, in general, there exists a large family of admissible power spectra which are all consistent.
Bounding the values of admissible spectral density functions is an ill-posed problem (see Section~\ref{sec:weakmetrics}).
Instead, the natural way to quantify power spectral uncertainty
is by bounding the power on (measurable) subsets of the frequency band. Therefore, the goal of this paper is to consider
the appropriate topology--the so-called weak topology, and to develop suitable metrics that can be used to quantify and measure
power spectral uncertainty.

Throughout, we consider stochastic processes
$\{y_t\,: \, t\in \mZ\}$ which are discrete-time, zero-mean, and second-order  stationary.
A typical set of statistics for a stationary stochastic process is a finite set of covariance (or, autocorrelation) samples.
The covariance samples
\[
c_k:=\cE\{y_t\bar{y}_{t-k}\}, \mbox{ for }k=0, \pm 1, \pm 2,\ldots,\pm n,
\]
where $\cE\{\cdot\}$ denotes the expectation operator,
provide moment constraints for the power spectrum $d\mu$ of the process:
\begin{equation}\label{moments}
c_k=\frac{1}{2\pi}\int_{-\pi}^\pi e^{-ik\theta}d\mu(\theta)\mbox{ for }k=0, \pm 1, \pm 2 \ldots, \pm n.
\end{equation}
The power spectrum is thought of as a non-negative measure on the unit circle
$\mT=\{z=e^{i\theta}\,:\,\theta\in(-\pi,\pi]\}$ (for notational simplicity also identified with the interval $(-\pi,\pi]$).
We use the symbol $\fM$ to denote the class of such measures and the problem of determining $d\mu\in\fM$
from the covariance samples (finitely or infinitely many) is known as
the trigonometric moment problem. Classical theory on this problem
originates in the work of Toeplitz and Carath\'eodory at the turn of
the 20$^{\rm th}$ century and has evolved into a rather deep chapter
of functional analysis and of operator theory
\cite{Akhiezer,KreinNudelman,Ger,foiasfrazho,BallGohbergRodman}.  The
classical monograph by Geronimus \cite{Ger} contains a wide range of
results on the trigonometric moment problem, the asymptotic behavior of
solutions, spectral factors and optimal predictors,
as well as explicit expressions for spectral envelops \cite[Theorem $5.7$]{Ger} (c.f. \cite{Capon,Haykin,Geo2001}). 
A more general form in which statistics may be available is when these represent the state covariance, or the output covariance, of a dynamical system driven by the stochastic process of interest. Such a dynamical system may represent a model of  physical processing (bandpass filtering at sensor locations, losses, structure of sensor array, etc.) or of virtual processing (software-based) of the original time-series. Either way, covariance statistics represent (generalized) moments of the power spectrum and a theory which is completely analogous to the theory of the trigonometric moment problem is available and provides similar conclusions, see \cite{ BGL1, BGL2, Geo2001,Geor00,statecov2}. In fact, the use of generalized statistics, which relates to beamspace processing, was explored in \cite{BGL2,Geor00} as a way to improve resolution in power spectral estimation over selected frequency bands. More recent work addresses spectral estimation with priors, computational issues, as well as important multivariate generalizations
\cite{Avventi2011,BlomqvistLindquistNagamune2003,FerrantePavonRamponi2008,
FerrantePavonZorzi2012,FerranteRamponiTicozzi2011,CFP,IT,
GL,GeorgiouLindquist2008,Nagamune2003,PavonFerrante2006,
RamponiFerrantePavon2009, RamponiFerrantePavon2010}.

The framework of the present work involves such moment problems specified by covariance statistics.
Invariably, moment statistics are estimated from a finite observation record and are known with
limited accuracy. Thus, in a typical experiment, as the observation
record of a time-series increases so does the accuracy and the length of the estimated partial 
covariance sequence. Our goal is to develop metrics that can be used to quantify spectral
uncertainty. More specifically, phrased in the context of the trigonometric moment problem,
we seek metrics between power spectra that have the following properties:
\begin{enumerate}
\item[(i)] given a finite set of covariance samples, the family of consistent power spectra has a finite diameter, and
\item[(ii)] the diameter of the uncertain set of power spectra shrinks to zero as both, the accuracy of the covariance samples increases and their number tends to infinity.
\end{enumerate}
The latter condition is dictated by the fact that the trigonometric
moment problem is known to be determined, i.e., there is a unique
power spectrum which is consistent with an infinite sequence of covariances. As
we will explain below (in Section~\ref{sec:main}), the proper topology which allows for these properties
to hold is the weak topology on measures (cf., \cite[page
$8$]{hoffman}). There is a variety of metrics that can be used to
metrize this topology, and thus, in principle, to quantify spectral
uncertainty. A contribution of this work is to suggest a class of
metrics for which the radius of spectral uncertainty and {\it a priori bounds}  are computable
given a finite set of (error-free) statistics.

In Section~\ref{sec:background} we review the trigonometric moment
problem and relevant concepts in functional analysis. In Section~\ref{sec:main} we define power-spectral uncertainty sets and discuss the relevance of weakly continuous metrics. In Section \ref{sec:weakmetrics} we present a collection of
weakly continuous metrics that, in different ways, are suitable for metrizing the space of power spectra.  In Section \ref{thecaseof} we compute
the diameter of uncertainty sets, for a particular choice of a metric, and elaborate on the 
limit properties of this uncertainty quantification. In Section \ref{sec:apps} we present an example that elucidates the relevance and applicability of the results
in practice. In Section \ref{sec:generalized} we explain how the framework applies  in the context of generalized statistics. In Section \ref{sec:THREEex} we highlight the use of this quantification of uncertainty in filter design ---we show
how to tune a filter-bank so as limit spectral uncertainty over some frequency range of interest. In the concluding section (Section \ref{sec:conclutions}) we summarize the results and outline possible future directions.

\section{The trigonometric moment problem, spectral representations, and weak convergence}\label{sec:background}

The covariances $c_k$, $k=0,\pm 1,\pm 2, \ldots$, of a stationary random process $\{y_t\,:\,t\in \mZ\}$ are the Fourier coefficients of the spectral measure $d\mu$ as in \eqref{moments}. These are characterized by the non-negativity of the Toeplitz matrices \cite{GrenanderSzego,Haykin}
\[\label{eq:Tn}
T_n=\left[\begin{array}{llll}
c_0&c_{-1}&\cdots&c_{-n}\\
c_1&c_0 &\cdots& c_{-n+1}\\
\vdots&\vdots &\ddots&\vdots\\
c_n&c_{n-1}&\cdots&c_0 \end{array}\right],
\]
for $n=0,1,\ldots$.  When $T_n> 0$ for $n\leq k$ and singular for $n=k+1$, then it is also singular for all $n>k$ and ${\rm rank}(T_{k+\ell})={\rm rank}(T_{k})=k+1$ for all $\ell\geq 1$.
In this case, $d\mu$ is singular with respect to the Lebesgue measure and consists of  finitely many ``spectral lines,'' equal in number to ${\rm rank}(T_{n})$ \cite[page 148]{GrenanderSzego}. Because $d\mu$ is a real measure,
$c_k=\bar{c}_{-k}$ for $k=0,\,1,\ldots$, hence we use only positive indices and refer
by
\[\bczn:=(c_0,\,c_1,\ldots,\,c_n)\]
to the vector of the first $(n+1)$ moments,  and by
\[\bc:=(c_0,\,c_1,\ldots)\] 
to the infinite sequence.
The sequence $\bc$ is said to be {\em positive} if $T_n >0$  for all $n$. Similarly $\bczn$ is said to be {\em positive} if $T_n>0$. Accordingly, the term {\em non-negative} is used when the relevant Toeplitz matrices are non-negative definite.

As noted in the introduction, the power spectrum of a discrete-time stationary process is a bounded non-negative measure on the unit circle. The derivative (of its absolutely continuous part) is referred to as the spectral density function, while the singular part typically contains jumps (spectral lines) associated with the presence of sinusoidal components. In general, the singular part may have a more complicated mathematical structure that allocates ``energy'' on a set of measure zero without the need for distinct spectral lines \cite[page 5]{GrenanderSzego}. From a mathematical viewpoint such spectra are important as they represent limits of more palatable spectra, and hence, represent a form of completion.

The natural topology where such limits ought to be considered is the so-called 
{\em weak topology}. This topology is also known as the weak$^*$ topology in functional analysis--a term which is less frequently used in the context of measures.
The weak topology is defined in terms of convergence of linear functionals and is explained next. We denote by $C(\mT)$ the class of real-valued continuous functions on
$\mT$. It is quite standard that the 
space of bounded linear functionals $\Lambda:C(\mT)\rightarrow\mR$,
can be identified with the space of bounded measures on $\mT$
\cite[page $7$]{hoffman}.  
More specifically, any bounded functional $\Lambda$ can be represented in the form
\[
\Lambda(f)=\int_\mT f(t)d\mu(t) \mbox{ for all }f\in C(\mT),
\]
with $d\mu$ being the corresponding measure--this is
the Riesz representation theorem.
Continuous functions now serve as ``test functions'' to differentiate between measures. Bounds on the corresponding integrals define the weak topology: a sequence of measures $d\mu_n$, $n=1,2,\ldots$, converges to $d\mu$ in the weak topology if $\int fd\mu_n \to \int fd\mu$ for every $f\in C(\mT)$.
Thus, for any two measures that are different, there exists a continuous function that the two measures integrate to different values.
In this setting, a measure can be specified uniquely by its Fourier coefficients. In fact,
given a positive sequence $\bc$, the unique corresponding measure $d\mu$ can be determined as the limit in the weak topology of finite Fourier sums or Cesaro means \cite[page 24]{hoffman}.

Non-negative measures are naturally associated with analytic and harmonic functions---a connection which has been exploited in classical circuit theory in the context of passivity. Herglotz' theorem \cite{Akhiezer} states that if $d\mu$ is a bounded non-negative measure on $\mT$, then
\begin{equation}\label{eq:H}
H[d\mu](z)=\frac{1}{2\pi}\int_{-\pi}^\pi \frac{e^{i\theta}+z}{e^{i\theta}-z}d\mu(\theta)
\end{equation}
is analytic in $\mD:=\{z\,:\,|z|<1\}$ and the real part is non-negative. Such functions are referred to as either ``positive-real'' or, as Carath\`{e}odory functions. Conversely, any positive-real function can be represented (modulo an imaginary constant) by the above formula for a suitable non-negative measure.
The Poisson integral of a non-negative measure $d\mu$
\begin{equation}\label{eq:P}
P[d\mu](z):=\frac{1}{2\pi}\int_{-\pi}^\pi P_r(t-\theta)d\mu(\theta), \quad z=re^{it},
\end{equation}
where $P_r(\theta)=\frac{1-r^2}{|1-re^{i\theta}|^2}$ is the Poisson kernel, is a harmonic function which is non-negative in $\mD$ and is equal to the real part of $H[d\mu](z)$. 
Given either a positive-real function $H(z)$, or its real part $P(z)$, the measure $d\mu$  such that $H(z)=H[d\mu](z)$ and $P(z)=P[d\mu](z)$ is uniquely determined by the limit of $P(re^{i\theta})d\theta\rightarrow d\mu$ as $r\rightarrow 1$ in the weak topology \cite[page 33]{hoffman}.
Thus, power spectra are, in a very precise sense, boundary limits of the (harmonic) real parts of positive-real functions.
 
\section{Uncertainty of spectral estimates}\label{sec:main}
We postulate a situation where covariances $\bczn$ are estimated from sample of a stochastic process $\{y_t\}_{t\in\mZ}$ with power spectrum $d\nu$, and where the estimation error in the entries of $\bczn$ are bounded by $\epsilon$\footnote{The more realistic situation, where the confidence intervals degrade with the order of covariance lags, can be dealt with in a similar manner, albeit with a bit more cumbersome notation.}. 
Thus, the ``true'' spectrum $d\nu$ belongs to the uncertainty set
\[
\cF_{\bczn,\epsilon}:=\left\{d\mu\ge 0: \left|c_k-\int_{-\pi}^\pi \!e^{-ik\theta} d\mu\right|<\epsilon, k=0,1,\dots,n\right\}.
\]
Likewise, any choice for a ``nominal'' spectrum $d\hat\nu$ consistent with our assumptions will also belong to $\cF_{\bczn,\epsilon}$.
Therefore, the distance between the two will be bounded by
the diameter of the uncertainty set,
\[
\rho_{\delta}(\cF_{\bczn,\epsilon}):=\sup\{\delta(d\mu_0,d\mu_1)\,:\, d\mu_0,d\mu_1 \in \cF_{\bczn,\epsilon}\},
\]
where $\delta$ is a suitable metric at hand.
Thus, our goal in this paper is to seek metrics $\delta$ on the space of positive measures $\fM$ that provide a meaningful and computationally tractable notion of a diameter for $\cF_{\bczn,\epsilon}$ thereby quantifying modeling uncertainty in the spectral domain.
To narrow down the search for suitable metrics, consider the scenario when the length of the data increases, and hence the accuracy as well as the number of covariance lags increases. In the limit, as the estimation error goes to zero and the number $n$ of covariance lags goes to infinity, the uncertainty set shrinks to the singleton 
\[
\{d\nu\}=\bigcap_{n\in\mN}\cF_{\bczn,\epsilon_n}, \quad(\epsilon_n\to 0 \mbox{ as } n\to\infty).
\]  
This is due to the fact that an infinite limit sequence $\bc$ defines a unique power spectrum--the trigonometric problem is determinate. The diameter should reflect this shrinkage to a singleton and tend to zero.
For this to happen, the underlying metric needs to be weakly continuous as stated next.

\begin{thm}\label{thm:weakstardistance}
Let $\delta$ be a metric on $\fM$.
Then
\begin{equation}\label{eq:weaklimit}
\rho_\delta(\cF_{\bczn,\epsilon_n} ) \to 0 \mbox{ as } \epsilon_n\to 0 \mbox{ and } n\to\infty,
\end{equation}
for every covariance sequence $\bc$ if and only if $\delta$ is weakly continuous.
\end{thm}
\begin{proof}
This can be seen by comparing the definition of $\cF_{\bczn,\epsilon_n}$ with the definition of open sets in the weak topology. See the appendix for a detailed proof.
\end{proof}

\begin{remark} Occasionally one may have additional a priori knowledge on the structure and smoothness of the power spectrum which would further limit the uncertainty set.
Quantifying such ``structured'' uncertainty would necessarily be problem-specific and is not considered in the present work. Instead, we take a viewpoint that allows
comparing power spectra in a unified way, regardless smoothness, presence of spectral lines, or membership in a specific class of models. \hfill $\Box$
\end{remark}

We now consider the case where the finite covariance sample $\bczn$ is known exactly. If $\bczn$ is  positive, then the uncertainty set
\[
{\cal F}_\bczn :=\left\{d\mu\ge 0: c_k=\int_{-\pi}^\pi \!e^{-ik\theta} d\mu, k=0,1,\dots,n\right\}
\]
contains infinitely many power spectra. If $\bczn$ is only non-negative, and hence $T_n$ is singular, then
the family ${\cal F}_\bczn$ consists of the single power spectrum $d\nu$ \cite[page 148]{GrenanderSzego}.
The following two results are immediate corollaries of Theorem~\ref{thm:weakstardistance}. The first one treats the case where the number of covariance lags goes to infinity, while the second, treats the case where the values of the covariance lags tend to those of a singular sequence. In both cases the diameter of the uncertainty set necessarily goes to zero for a weakly continuous metric.
\begin{cor} \label{prp:limit1}
Let $\bc$ be a non-negative sequence and let  $\delta$ be a weakly continuous metric. Then
\[ \rho_\delta(\cF_\bczn)\to 0, \mbox{ as } n\to \infty.
\]
\end{cor}
\begin{proof} This follows directly from Theorem~\ref{thm:weakstardistance} and by noting that 
\[
\cF_{\bczn}\subset\cF_{\bczn,\epsilon}
\]
for any $\epsilon>0$.
It also follows from \cite[\S 1.16]{Ger} in view of
Proposition \ref{prp:eqlimits} in Section~\ref{ssec:poisson} below.
\end{proof}

\begin{cor}  \label{prp:limit2}
Let ${\bc}_{0:n}$ be a vector of covariance lags such that the corresponding $T_n$ is a singular Toeplitz matrix, and let
$\hat{\bc}_{0:n}(k)$ ($k=1,2,\ldots$) be a sequence of vectors of covariance lags tending to ${\bc}_{0:n}$. If 
$\delta$ is a weakly continuous metric then
\[\rho_\delta(\cF_{\hat{\bc}_{0:n}(k)})\to 0, \mbox{ as } k\to\infty.
\]
\end{cor}
\begin{proof}
Follows directly from Theorem \ref{thm:weakstardistance}. See also \cite{KG} for an independent detailed argument.
\end{proof}

\begin{remark}\label{rm:TV}
It should noted that the total variation ($\int|d\mu_0-d\mu_1|$) is not weakly continuous and therefore the conclusions of the two corollaries would fail if this was used as the metric. To see this, note that if $\bczn$ is positive, then $\cF_\bczn$ contains infinitely many measures and among them at 
least two singular measures with non-overlapping support, i.e., $\supp(d\mu_0)\cap\supp(d\mu_1)=\emptyset$ (e.g., see \cite{KreinNudelman}). Then the total variation of their difference is always $2c_0$. \hfill $\Box$
\end{remark}

\section{Weakly continuous metrics}\label{sec:weakmetrics}
In general, a finite set of second-order statistics cannot dictate the precise value of the power spectrum locally. Indeed, given any finite positive sequence  $\bczn$ and any $\theta_0\in (-\pi,\pi]$, then for any value $\alpha\ge 0$ there exists  an $\epsilon>0$ and an absolutely continuous measure $d\mu=fd\theta \in \cF_\bczn$ such that 
\[
f(\theta)=\alpha\mbox{ for } \theta\in(\theta_0-\epsilon,\theta_0+\epsilon).
\]
What can be said instead, is that the range of values
\begin{equation}\label{eq:local}
\left\{\int_\mT gd\mu\,:\, d\mu\in \cF_\bczn\right\},
\end{equation}
for any particular test function $g\in C(\mT)$, is bounded.
Furthermore, as $n\to \infty$, this range tends to zero.
In fact, due to weak continuity, 
the range of values tend to zero for any of the scenarios in Theorem \ref{thm:weakstardistance} and its two  corollaries. Finding the maximum and the minimum of \eqref{eq:local} is a linear programming problem on an infinite dimensional domain. 
Provided $g$ is symmetric real and the covariance sequence $\bczn$ is real, the dual problems, which give the lower and upper bounds of \eqref{eq:local}, are
\begin{align}
&\max \left\{\lambda \bczn^T \,:\, \sum_{k=0}^n \lambda_k \cos(k\theta)\le g(\theta) , \theta\in(-\pi,\pi]\right\},\label{eq:LinEq1}\\
&\min \left\{\lambda \bczn^T \,:\,  g(\theta) \le\sum_{k=0}^n \lambda_k \cos(k\theta), \theta\in(-\pi,\pi]\right\},\label{eq:LinEq2}
\end{align}
where $\lambda=(\lambda_0,\,\lambda_1,\, \ldots,\,\lambda_n)$ are Lagrange multipliers. 

\begin{remark}
Along these lines Lang and Marzetta in \cite{LangMarzetta1983, MarzettaLang1984} sought to quantify the maximal and minimal spectral mass in a specified interval given the covariances $\bczn$. To this end we may take $g=\chi_I$ the characteristic function of an interval $I$, that is, $\chi_I(\theta)=1$ if $\theta\in I$ and $0$ otherwise. Lower and upper bounds on $\int_I d\mu$ are finite and are then given by \eqref{eq:LinEq1} and \eqref{eq:LinEq2}, respectively. However, since $g=\chi_I$ is not continuous, the mass in an interval is not a weakly continuous quantity, and the requirements in  Corollary~\ref{prp:limit1} does not hold. In fact, for this case the gap between the upper and lower bound does not necessarily converge to zero as $n$ goes to infinity. This occurs, e.g., in the case when the true spectrum has a spectral line at an end point of the interval. \hfill $\Box$
\end{remark}

 A class of weakly continuous metrics can be sought in the form
\begin{equation}\label{eq:locmetric}
\delta(d\mu_0, d\mu_1)=\sup_{\xi\in K}\left|\int_\mT g_\xi(d\mu_0-d\mu_1)\right|,
\end{equation}
for  $\{g_\xi\}_{\xi\in K}\subset C(\mT)$, provided the family $\{g_\xi\}_{\xi\in K}$ of test functions is sufficiently rich to distinguish between measures and yet, small enough so that continuity is ensured. The precise conditions are given next. 

\begin{prop}\label{prp:locmetric}
The functional $\delta(d\mu_0, d\mu_1)$ defined in \eqref{eq:locmetric} is a weakly continuous metric if and only if the following two conditions hold:
\begin{enumerate}
\item[(a)] for any two  measures $d\mu_0,d\mu_1\in \fM$, there is a $\xi\in K$ such that $\int_\mT g_\xi d\mu_0\neq\int_\mT g_\xi d\mu_1$, and
\item[(b)] the set $\{g_\xi\}_{\xi\in K}$ in $C(\mT)$ is equicontinuous\footnote{A family of functions $\{g_\xi\}_{\xi\in K}\subset C(\mT)$ is said to be equicontinuous if for any $\epsilon$ there exists a $\gamma$ such that $|g_\xi(\theta_1)-g_\xi(\theta_2)|<\epsilon$ if $|\theta_1-\theta_2|<\gamma$ for all $\theta_1,\theta_2\in\mT,$ and $\xi\in K$.} and uniformly bounded.
\end{enumerate}
\end{prop}
\begin{proof}
See the appendix.  
\end{proof}
In essence, condition (a) ensures positivity while condition (b) ensures weak continuity.
The triangle inequality and symmetry always hold for such $\delta$. The total variation norm is an example of why (b) is needed---it is a norm  of the form \eqref{eq:locmetric} where the set of test function are  the $C(\mT)$ unit ball, $\{g: \|g\|_\infty\le 1\}$, but it is not weakly continuous. 
This is due to the fact that the unit ball in  $C(\mT)$ is not equicontinuous.

\begin{remark}
A more general family of distances are of the form
\[
\delta(d\mu_0,d\mu_1)=\sup_{\scriptsize\begin{array}{c}g_0(\theta)\in K_0, g_1(\phi)\in K_1,\\ g_0(\theta)+g_1(\phi)\in K\end{array}} \int_\mT g_0d\mu_0+\int_\mT g_1d\mu_1
\]
where $K_0, K_1\subset C(\mT)$ and $K\subset C(\mT\times\mT)$. By selecting the sets $K_0, K_1$, and $K$ properly, $\delta$ (or a monotone function of $\delta$) will be a weakly continuous metric. One such example is the metrics based on optimal transportation treated in \cite{GKT}, where the metrics have non-local properties such as geodesics which preserve lumpedness. \hfill $\Box$
\end{remark}

Next we consider three ways for devising weakly continuous metrics. The first uses smoothing of power spectra to be compared by suitable test functions in a way that is analogous to the use of classical window kernels in periodogram estimation \cite{StoicaMoses}. The second is based on Monge-Kantorovich optimal mass transportation where a cost is associated with mismatch in the frequency range where power resides. In this geometry, optimal-transport geodesics may be used to model slow time-varying drift in the spectral power of non-stationary time-series \cite{GKT} --- such models for non-stationarity lessen artifacts present when using ordinary interpolation (e.g., fade-in fade-out \cite{Jiang2}).
  The third is based on  Poisson kernels and is more suitable for differentiating spectra based on their content on specified frequency bands. The connection between Poisson kernels and the analytic and harmonic functions in \eqref{eq:H} and \eqref{eq:P} allows for evaluating bounds and the diameter of the uncertainty set with respect to the corresponding distances. This will be explored in the case where finitely many error-free covariances are known in Sections \ref{thecaseof} to \ref{sec:THREEex}.

\subsection{Metrics based on smoothing}

A simple way to devise weakly continuous metrics which has a classical flavor is to
first smoothing the measures via convolution with a fixed suitable
continuous function, and then to compare the smoothed spectral densities. This echoes the use of windowing Fourier techniques in the time domain \cite{StoicaMoses} where a  suitable choice of a window is used to trade-off resolution and variance of the estimator. Likewise here, the choice of a windowing function determines the resolution of the metric.  

Thus, let $g\in C(\mT)$ be such a windowing function, and define
\[
\delta_{{\rm smooth},g}(d\mu_0,d\mu_1):=\|g\ast(d\mu_0-d\mu_1)\|_{\infty}.
\] 
Here,
\[
(g\ast d\mu)(\xi)=\int_{-\pi}^{\pi} g(\xi-\theta)d\mu(\theta)
\]
denotes the circular convolution and $\|\cdot\|_\infty$ the $L_\infty$
norm. In the view of Proposition~\ref{prp:locmetric}, $\delta_{{\rm smooth},g}$ is of the form
\[
\|g\ast(d\mu_0-d\mu_1)\|_{\infty}=\sup_{\xi\in(-\pi,\pi]}\left|\int_{-\pi}^\pi g(\xi-\theta)(d\mu_0(\theta)-d\mu_1(\theta))\right|,
\]
and hence, condition (b) of the proposition holds. In addition, the chosen convolution-kernel functions must not have any zero Fourier coefficient, otherwise the approach will fail to differentiate between certain measures. To see this, let
$g(\theta)=\sum_{k=-\infty}^\infty g_k e^{ik\theta}$ and let $(\dots,
a_{-1},a_0,,a_1,\ldots)$ be the Fourier coefficients of
$d\mu_0(\theta)-d\mu_1(\theta)$, then
\[
g\ast (d\mu_0-d\mu_1)(\xi)=\sum_{k=-\infty}^\infty g_{-k} a_ke^{ik\xi}.
\]
If $g_k\neq 0$ for all $k\in\mZ$, the above expression cannot vanish identically unless all the $a_k$'s are zero, in which case 
$d\mu_0=d\mu_1$. In this case (a) holds and it follows from Proposition~\ref{prp:locmetric} that
$\delta_{{\rm smooth},g}(d\mu_0,d\mu_1)$ is a weakly continuous metric. This leads to the next proposition.
\begin{prop}
Let $g\in C(\mT)$ be a windowing function with non-vanishing Fourier coefficients. Then 
$\delta_{{\rm smooth},g}(d\mu_0,d\mu_1)$ is a weakly continuous metric. 
\end{prop}

\subsection{Metrics based on optimal transportation}
A rapidly growing literature \cite{Villani} on a classical problem, known as the Monge-Kantorovich transportation problem, has impacted a wide range of disciplines, from probability theory to fluid dynamics and economy \cite{RachevRuschendorf1998}.  Optimal transportation refers to the correspondence between distributions of masses that induce the least amount of transportation cost\footnote{L. Kantorovich received the $1975$ Nobel Prize for the impact of this theory on allocation of economic resources.}. The optimal transportation cost between two probability distributions induces weakly continuous metrics, known as Wasserstein metrics, which are extensively used in probability theory. In order to handle more general distributions we need a suitable modification to compare unequal masses. This we do next and connect with the formalism in \eqref{eq:locmetric}.

The Monge-Kantorovich transportation problem amounts to minimizing the cost of
transportation between two distributions of equal mass, e.g., $d\mu_0$
and $d\mu_1$ where $\int_\mT d\mu_0=\int_\mT d\mu_1$. In this, a
transportation plan $d\pi(\theta,\phi)$ is sought which corresponds to
a non-negative distribution on $\mT\times \mT$ and is such that
\begin{equation}\label{eq:transportpi}
\int_{\theta \in \mT} d\pi(\theta,\phi)=d\mu_0(\phi) \mbox{ and } \int_{\phi\in \mT} d\pi(\theta,\phi)=d\mu_1(\theta).
\end{equation}
Then, the minimal cost
\[
\min\left\{ \int_{\mT\times \mT} |\theta-\phi|d\pi(\theta,\phi)\,:\, d\pi \mbox{ satisfies } \eqref{eq:transportpi}\right\}
\]
is the Wasserstein-1 distance between $d\mu_0$ and $d\mu_1$, and is a
weakly continuous metric (see, e.g., \cite[chapter $7$]{Villani}).
This problem admits a dual formulation, known as the Kantorovich
duality:
\[
W_1(d\mu_0,d\mu_1)=\max_{\|g\|_L\le 1} \int g(d\mu_0-d\mu_1),
\]
where  $\|f\|_L=\sup_{\theta,\phi}\frac{|f(\theta)-f(\phi)|}{|\theta-\phi|}$ denotes the Lipschitz norm.

Power spectra, in general, cannot be expected to have the same total mass. In this case, $\delta_{1,\kappa}(d\mu_0,d\mu_1)$ defined by
\begin{equation}\label{eq:tildeT}
\inf_{\int d\nu_0=\int d\nu_1}\hspace*{-5pt}W_1(d\nu_0,d\nu_1)+\kappa\sum_{i=0}^1\int_\mT |d\mu_i-d\nu_i|,
\end{equation}
is a weakly continuous metric for an arbitrary but fixed
$\kappa>0$. The interpretation is that $d\mu_0$ and $d\mu_1$ are
perturbations of the two underlying measures $d\nu_0$ and $d\nu_1$, respectively, which have equal mass. Then, the cost of transporting $d\mu_0$ and $d\mu_1$
to one another can be thought of as the cost of transporting $d\nu_0$
and $d\nu_1$, to one another, plus the size of their respective
perturbations from $d\mu_0$ and $d\mu_1$.
This is introduced in \cite{GKT} and this metric admits a dual formulation
\[
\delta_{1,\kappa}(d\mu_0,d\mu_1)=\max_{\begin{array}{c}\|g\|_\infty\le
\kappa\\\|g\|_L\le 1
\end{array}} \int g(d\mu_0-d\mu_1),
\]
which is in the form of the Proposition~\ref{prp:locmetric}.
Various other generalizations of the transportation distance that apply to power spectra are also being proposed and studied in  \cite{GKT}.

\subsection{Metrics based on the Poisson kernel}\label{ssec:poisson}
Power spectra are weak limits of the real part of analytic functions on the unit disc, as indicated earlier. Comparison of these functions induces weakly continuous metrics which readily fall under the framework of \eqref{eq:locmetric}.  
Interestingly, this approach allows for both the computation of explicit/analytic bounds on uncertainty sets (see Section \ref{thecaseof}) and for specifying a  frequency dependent resolution of a metric (see Remark~\ref{rm:frdep} and the example in Section~\ref{sec:generalized}). 

Recall from Section~\ref{sec:background} that the harmonic function associated with a measure is the Poisson integral, defined as 
\[
P[d\mu](z)=\frac{1}{2\pi}\int_{-\pi}^\pi P_r(t-\theta)d\mu(\theta), \quad z=re^{it}.
\]
Weak convergence of measures is equivalent
to certain types of convergence of their harmonic counterpart.
\begin{prop} \label{prp:eqlimits}
Let $\{d\mu_k\}_{k=1}^\infty$ be a sequence of uniformly bounded signed measures on $\mT$, let $d\mu$ be a bounded measure on $\mT$, and let $u(z)=P[d\mu](z)$, $u_k(z)=P[d\mu_k](z)$ be their corresponding Poisson integrals.
The following statements are equivalent:
\begin{enumerate}
\item[(a)] $d\mu_k\rightarrow d\mu$ weakly,
\item[(b)] $u_k(z) \rightarrow u(z)$ pointwise $\forall z\in \mD$,
\item[(c)] $u_k(z) \rightarrow u(z)$ in $L_1(\mD)$,
\item[(d)] $u_k(z) \rightarrow u(z)$ uniformly on every compact subset of $\mD$.
\end{enumerate}
\end{prop}
\begin{proof}
The proof is given in the appendix.
\end{proof}
Each of the statements $(b), (c),$ and $(d)$ may be used for devising weakly continuous metrics. We shall focus on the statement $(d)$, indicating that weakly continuous metrics can be constructed by comparing the harmonic functions on a subset of $\mD$.  
In fact, the maximal distance between the harmonic functions on a closed non-finite set $K\subset\mD$  
gives rise to a weakly continuous
metric
\begin{equation}\label{distanceK}
\delta_{K}(d\mu_0,d\mu_1)=\max_{z\in K}|P(d\mu_0-d\mu_1)(z)|.
\end{equation}
This is true, since the resulting family of the Poisson kernels
satisfies the properties in Proposition~\ref{prp:locmetric}. To see this, first note that any two harmonic functions which coincides on $K$, a closed non-finite set inside $\mD$, must be identical, hence $(a)$ is satisfied. Further more, since $K\subset \gamma \mD$ for some $\gamma<1$, the magnitude and derivative of $P_r(t-\theta)$ is uniformly bounded when $re^{it}\in K$, hence 
$(b)$ holds.

\begin{remark} \label{rm:frdep} In practice, it is often the case that 
one is interested in comparing spectra over selected frequency bands.
To this end, various schemes have been considered which rely on
pre-processing with a choice of ``weighting'' filters and filter banks
(see e.g., \cite{Capon}, \cite{PPVaidyanathan}, and \cite{BGL1,statecov2}). The choice of the point-set $K$ in \eqref{distanceK} can
be used to dictate the resolution of the metric over such frequency
bands. To see how this can be done, consider $K$ to designate an arc
$\{\xi=re^{i\theta}\,:\,\theta\in[\theta_0-\epsilon,\theta_0+\epsilon]\}$. This
satisfies the conditions of Proposition~\ref{prp:locmetric} and thus,
$\delta_K$ is a weakly continuous metric. At the same time, the values
$P[d\mu](\xi)$, with $\xi\in K$, represent the variance at the output
of a filter with transfer function $z/(z-\xi)$. These are bandpass
filters with a center frequency $\arg(\xi)$ and bandwidth which
depends on the choice of $r$. Thus, in essence,
the metric compares the respective
variance after the spectra have been weighted by a continuum (for
$\xi\in K$) of such frequency-selective bank of filters. \hfill $\Box$
\end{remark}

\section{The size of the uncertainty set}\label{thecaseof}
The diameter of the uncertainty set with respect to the distance $\delta_K$ turns out to be especially easy to compute -- it is realized as the distance between 
two ``diametrically opposite'' measures with only $n+1$ spectral lines each
(i.e., measures having compact support). This is the content of the following proposition.

\begin{prop} \label{prp:supKnorm}
Let ${\bf c}_{0:n}$ be a positive covariance sequence and let $K\subset \mD$ be closed. Then 
\begin{align*}
&\rho_{\delta_K}(\cF_{{\bf c}_{0:n}})=\\ &\max_{z\in K}\left\{2\left(\left|\frac{\frac{2}{1-z\bar{z}}+\langle b_z,d_z\rangle_{T^{-1}}}{\langle b_z,b_z\rangle_{T^{-1}}}\right|^2-\frac{\langle d_z,d_z\rangle_{T^{-1}}}{\langle b_z,b_z\rangle_{T^{-1}}}\right)^{\frac{1}{2}}\right\},
\end{align*}
where
\[
b_z=\!\left(\!\begin{array}{l}
z^{-1}\\
z^{-2}\\
\vdots\\
z^{-n-1}
\end{array}\!\!\!\right),\,
d_z= \!\left(\!\begin{array}{l}
z^{-1}(c_0)\\
z^{-2}(c_0 +2c_1z)\\
\vdots\\
z^{-n-1}(c_0 +2c_1z+\cdots+2c_nz^n)
\end{array}\!\!\!\right),
\]
and $\langle x,y\rangle_{T^{-1}}$ denotes the inner product 
\[\langle x,y\rangle_{T^{-1}} := y^* T_n^{-1} x.
\]
Furthermore,
$\rho_{\delta_K}(\cF_{{\bf c}_{0:n}})$ is attained as the distance
between two elements of $\cF_{{\bf c}_{0:n}}$ which are both singular
with support containing at most $n+1$ points.
\end{prop}
\begin{proof}
The proof is given in the appendix.
\end{proof}

Both claims in Proposition \ref{prp:supKnorm} can be used separately
for computing $\rho_{\delta_K}(\cF_{{\bf c}_{0:n}})$. The first one
suggests finding a maximum of a real-valued function over $K$. The
second claim suggests a search for a maximum of
$\delta_K(d\mu_1,d\mu_2)$ over a rather small subset of
$\ext(\cF_\bczn)$, namely nonnegative sequences $\bc_{0:(n+1)}$
parametrized by $c_{n+1}$; i.e., solutions of the quadratic equation
\begin{equation}\label{eq:T0}
\det(T_{n+1})=0.
\end{equation}
The (complex) values for $c_{n+1}$ satisfying \eqref{eq:T0}, lie on a circle in the complex plane, and hence, computation of $\rho_{\delta_K}(\cF_{{\bf
c}_{0:n}})$ requires search on a torus (each of the two extremal
$d\mu_1$, $d\mu_2$ where the diameter is attained can be thought of as
points on the circle).

We elucidate this with an example.
Figure \ref{fig:ex3} shows $\rho_{\delta_K}(\cF_\bczn)$ for
\[
\bc_{0:2} =(1,\;c_1,\;c_2)
\]
as a function of the corresponding {\em partial 
autocorrelation coefficients}, also known as {\em Schur parameters} (see the appendix),
\begin{eqnarray*}
-1< &\gamma_1:=c_1&< 1,\\
-1\leq &\gamma_2:= \frac{\det\left( \begin{array}{cc} c_1 & c_2\\ 1 & c_1 \end{array}\right)}
{\det\left( \begin{array}{cc} 1& c_1\\ \bar{c}_1 & 1 \end{array}\right)} &\leq 1,
\end{eqnarray*}
and $K$ is taken as $\{z\;:\;|z|\leq 0.5\}\subset \mD$.

The plot confirms that the diameter decreases to zero as the
parameters or, alternatively, the covariances $c_1$ and $c_2$, tend to
the boundary of the ``positive'' region (which in the Schur
coordinates corresponds to the unit square). 
However, it is interesting to note that the diameter of $\cF_\bczn$ as a function of $\bczn$ has several local maxima. 
This maximal diameter may be explicitly calculated, hence provides an {\it a priori} bound on the uncertainty. 
\begin{thm}\label{thm:COVbound}
Let $r=\max(|z|\,:\, z\in K)$.
Then
\begin{equation}\label{eq:covbound}
\rho_{\delta_K}(\cF_\bczn)\le \frac{4c_0|r|^{n+1}}{1-|r|^2}. 
\end{equation}
Further, \eqref{eq:covbound} holds with equality if and only if $\bczn=(c_0, c_0\bar \alpha, c_0\bar \alpha^2,\ldots, c_0\bar \alpha^n )$ for some $\alpha\in K$ with 
 $|\alpha|=r$.
\end{thm}
\begin{proof}
The proof is given in the appendix.
\end{proof}
 
\begin{figure}
\centering
\includegraphics[scale=0.45]{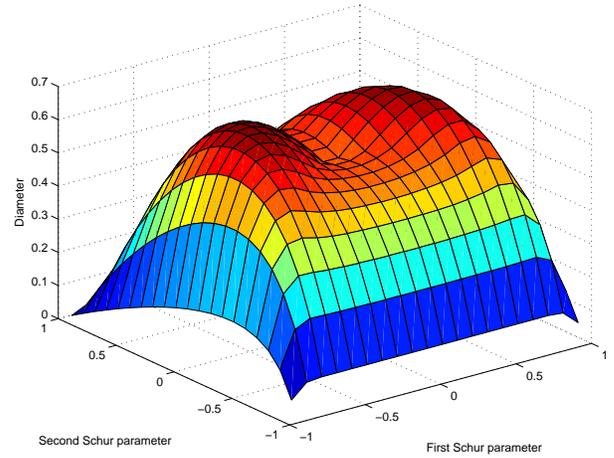}
\caption{{The uncertainty diameter $\rho_{\delta_K}$ as a function of $\gamma_1,\gamma_2$ when $c_0=1$ and $K=\{z:|z|\le 0.5\}.$}}
\label{fig:ex3}
\end{figure}

\begin{remark}
Computation of the diameter $\rho_\delta(\cF_\bczn)$ of the
uncertainty set amounts to solving the infinite-dimensional
optimization problem
\begin{equation}\label{eq:convexinf}
\sup \{\delta(d\mu_1,d\mu_2)\,:\,d\mu_1,d\mu_2\in\cF_\bczn\}.
\end{equation}
If $\delta$ is a weakly continuous and jointly convex function, then
the diameter is attained as the precise distance between two elements
which are extreme points $\cF_\bczn$. Extreme points are the points
with the property that they themselves are not a convex combination of
other elements in the set; the set of extreme points is denoted by $\ext(\cdot)$.  Then, $d\mu \in
\ext(\cF_\bczn)$ if and only if $d\mu\in\cF_\bczn$ and the support of
$d\mu$ consists of at most $2n+1$ points (see \cite{KG}).  Thus,
$\ext(\cF_\bczn)$ admits a finite dimensional characterization and
\eqref{eq:convexinf} reduces to a finite dimensional problem.\hfill $\Box$
\end{remark}

\section{Identification in a weak sense}\label{sec:apps}
In this section we elucidate how the uncertainty set is affected by the number of moments and show that spectra may be close in the weak sense even though they are qualitatively very different.

Consider the stochastic process 
\[
y_t=\cos(0.5t+\varphi_1)+\cos(t+\varphi_2)+w_t+\frac{1}{3}w_{t-1}
\]
where $w_t$ is a white noise process and $\varphi_1,\varphi_2$ are
random variables with uniform distribution on $(-\pi,\pi]$. The power
spectrum $d\nu$ is depicted in Figure~\ref{fig:truespec} and the
spectrum has both an absolutely continuous part as well as a singular
part. We would like to identify this spectrum relying on covariance data and derive 
bounds on the estimation error. We will use the
metric $\delta_{K}$ where $K=\{z\,:\,|z|=0.9\}$, i.e.,
\[
\delta_K(d\mu_0,d\mu_1)=\sup_{|z|= 0.9}|P(d\mu_0-d\mu_1)(z)|.
\]

\begin{figure}
\centering
\includegraphics[scale=0.45]{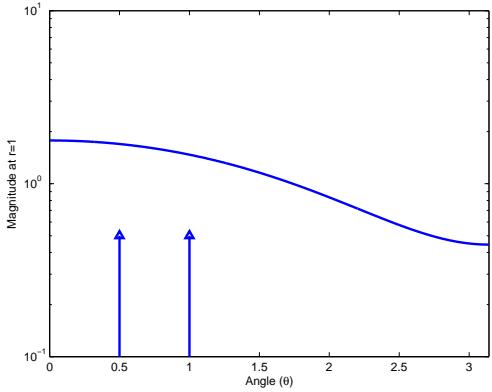}
\caption{{The ``true'' power spectrum $d\nu$.}}
\label{fig:truespec}
\end{figure}

Let $\bc$ be the covariance sequence of $d\nu$ and let $d\mu_5$ and
$d\mu_{20}$ be the power spectra with highest entropy in the sets
$\cF_{\bc_{0:5}}$ and $\cF_{\bc_{0:20}}$,
respectively. Figure~\ref{fig:deg5} compares $d\mu_5$ and
$d\nu$ where the estimation error and the uncertainty diameter are
\[ 
\delta_K(d\nu,d\mu_5)=5.66,\quad \rho_{\delta_K}(\cF_{{\bf c}_{0:5}})=20.79. 
\]
The first subplot shows and compares these two power spectra.
The second subplot
displays  $P[d\nu](0.9e^{i\theta})$, $P[d\mu_5](0.9e^{i\theta})$,
along with bounds on $P[d\mu](0.9e^{i\theta})$ when $\mu\in \cF_{\bc_{0:5}}$. It is seen that the
spectrum $d\mu_5$ does not distinguish the two peaks. 
In order to distinguish the two spectral lines, the information in $\bc_{0:5}$
is clearly not sufficient as the $\delta_K$-bounds are  substantial.

\begin{figure}
\centering
\includegraphics[scale=0.45]{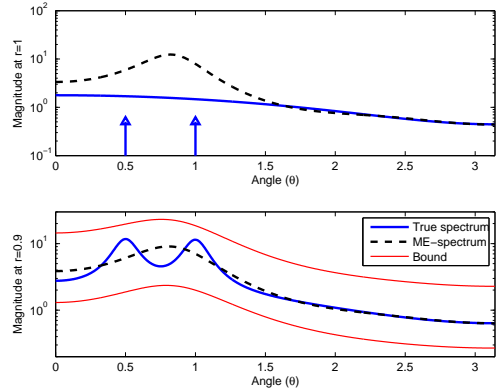}
\caption{{Subplot 1: The power spectrum $d\nu$ (solid),  $d\mu_5$ (dashed). 
Subplot~2:  $P[d\mu](0.9e^{i\theta})$ (solid), $P[d\mu_5](0.9e^{i\theta})$ (dashed),
along with bounds based on $\bc_{0:5}$.}}
\label{fig:deg5}
\end{figure}

Figure~\ref{fig:deg20} now compares $d\mu_{20}$ and $d\nu$ in a similar manner. 
The estimation error and the uncertainty diameter are
\[ 
\delta_K(d\nu,d\mu_{20})=0.29,\quad \rho_{\delta_K}(\cF_{{\bf c}_{0:20}})=2.52. 
\]
Here, $d\mu_{20}$ has two peaks close to the spectral lines and
$P[d\mu_{20}](0.9e^{i\theta})$ resembles $P[d\nu](0.9e^{i\theta})$ quite
closely. In fact, the bounds/envelops already reflect the presence of the two peaks.

To amplify the point made above, consider $d\mu_{\rm line}$ to be the (unique) power
spectrum in $\cF_{\bc_{0:20}}$ having Schur parameter
$\gamma_{21}=1$; this corresponds to a deterministic process (having only spectral lines) and is depicted in
Figure~\ref{fig:detspec}.
Subplot $2$ shows $P[d\mu_{\rm line}](0.9e^{i\theta})$ and how it ``sits'' within the respective bounds. In the absence of additional information, $d\mu_{\rm line}$, $d\mu_{20}$, or any other power spectrum in $\cF_{\bc_{0:20}}$ is admissible. The ``worst case'' distance between any two is the diameter computed above.

\begin{remark} Even though the three spectra $d\nu$, $d\mu_{20}$, and $d\mu_{\rm line}$, have identical covariances $\bc_{0:20}$, they are quite
different in terms of their respective singular and continuous
parts. However, they are similar in their distribution of spectral-mass -- they have most of their mass located around the frequency points $\theta=0.5$ and $\theta=1$, and this is what the weak topology captures.
\end{remark}

\begin{remark} Standard pointwise distances between $d\nu$, $d\mu_{20}$, and $d\mu_{\rm line}$ do not provide a meaningful comparison. For instance, the Itakura-Saito distance \cite{GrayBuzoGray1980}, the Kullback-Leibler divergence \cite{GL}, and the Cepstral distance \cite{GrayMarkel1976}, because they contain a logarithmic term, give the value of $\infty$ when comparing $d\mu_{20}$ and $d\mu_{\rm line}$. On the other hand, the $L_2$ metric does not apply to the present context because spectral lines cannot be viewed as ``$L_2$ functions'' and if approximated the norm diverges to infinity. Finally, the total variation 
does not differentiate when spectral lines are nearby or far apart
(c.f., Remark~\ref{rm:TV}).
\end{remark}

\begin{figure}
\centering
\includegraphics[scale=0.45]{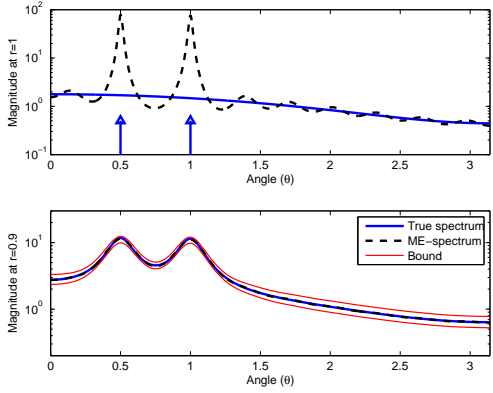}
\caption{{Subplot 1: Power spectrum $d\nu$ (solid), $d\mu_{20}$ (dashed). 
 Subplot~2:  $P[d\mu](0.9e^{i\theta})$ for the true spectrum (solid), $P[d\mu_{20}](0.9e^{i\theta})$ for $d\mu_{20}$ (dashed),
along with bounds based on $\bc_{0:20}$.}}
\label{fig:deg20}
\end{figure}

\begin{figure}
\centering
\includegraphics[scale=0.45]{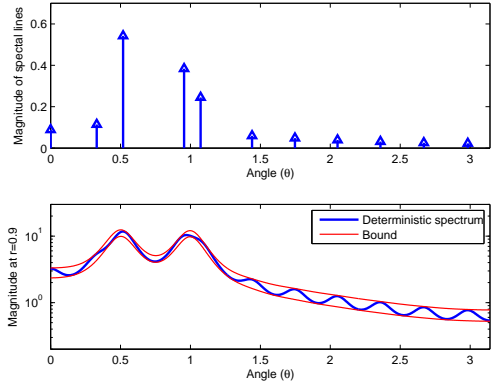}
\caption{{Subplot 1: Power spectrum $d\mu_{\rm line}$ (line spectrum).  Subplot~2: $P[d\mu_{\rm line}](0.9e^{i\theta})$ for the line spectrum, along with bounds based on $\bc_{0:20}$.}}
\label{fig:detspec}
\end{figure}

\section{Generalized statistics}\label{sec:generalized}

Our analysis extends readily to the case of generalized statistics \cite{BGL2,Geor00,BGL1,Geo2001}.
The formalism in these references, nicknamed THREE (for ``tunable high resolution estimation'') allows for the possibility of tunable filter-banks and was shown to provide improved resolution, albeit, quantitative assessments of the benefits exist only in special cases \cite{amini}. We briefly sketch the formalism here, for lack of space, and we refer to the aforementioned references for more detailed accounts.

\begin{figure}[htb]
\begin{center}
\setlength{\unitlength}{.01in}
\parbox{304\unitlength}
{\begin{picture}(304,150)
\thicklines
\put(120,105){\framebox(50,30){$G_0(z)$}}
\put(120,65){\framebox(50,30){$G_1(z)$}}
\put(120,15){\framebox(50,30){$G_n(z)$}}

\put(90,120){\vector(1,0){30}}
\put(90,80){\vector(1,0){30}}
\put(90,30){\vector(1,0){30}}

\put(90,30){\line(0,1){90}}
\put(80,75){\line(1,0){10}}

\put(170,120){\vector(1,0){30}}
\put(170,80){\vector(1,0){30}}
\put(170,30){\vector(1,0){30}}

\put(145,60){\circle{1}}
\put(145,55){\circle{1}}
\put(145,50){\circle{1}}

\put(185,130){\makebox(0,0){{\large $u_0$}}}
\put(185,90){\makebox(0,0){{\large $u_1$}}}
\put(185,40){\makebox(0,0){{\large $u_n$}}}

\put(70,75){\makebox(0,0){{\large $y$}}}
\end{picture}
}
\end{center}
\caption{Bank of filters.}\label{filterbank}
\end{figure}
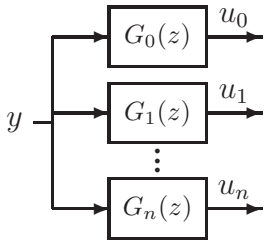

We explain the formalism of generalized statistics in the setting of ``filter-banks'', i.e., we consider the stochastic process $y_t$ as driving a bank of first-order dynamical systems with transfer functions
\[
G_k(z):=\frac{z}{z-z_k},\mbox{ for } k=0,1,\ldots,n, \mbox{ with }|z_k|<1
\]
as shown in Figure \ref{filterbank}.
The joint covariance matrix of the filter-bank outputs is
\[
P=E\{{\bf u}(t){\bf u}(t)^*\},
\]
where ${\bf u}:=(u_0(t),\,u_1(t),\,\ldots,u_n(t))^T$ . As indicated earlier $t\in\mZ$ is the time index.
The covariance matrix takes the form of a Pick matrix
\begin{equation}\label{eq:pick}
P:=\left[\frac{w_k+\bar w_\ell}{1-z_k\bar z_\ell}\right]_{k,\ell=0}^n
\end{equation}
where
\[
w_k=\frac{1}{2}(1-z_k^2)\cE\{u_k^2\}
\]
(see \cite[Equations (2.8), (2.10)]{BGL2} and \cite[page 783, Equation (7)]{Geor00}).
The matrix $P$ replaces the ordinary Toeplitz covariance in the previous sections.
Certain observations are in place: given the filter-bank dynamics, i.e.,  the $z_k$'s,  i) $P$ depends only on the values $w_k$, and ii) the cross-covariances between filter-bank elements can be computed from the output covariances of all elements individually, that is, from the $w_k$'s.

A rather complete theory has been developed to characterize power spectra for the input process that are consistent with output-covariance (more generally, state-covariance) statistics. This theory provides among other things a construction of the unique input spectrum of maximal entropy, spectral envelops that are reminiscent of the Capon pseudo-spectra, and the identification of spectral lines with techniques analogous to the theory of the Pisarenko Harmonic Decomposition, MUSIC, ESPRIT, etc., and has been worked out in detail for matrix-valued power spectra as well  (see e.g., \cite{Geor00,statecov2,CFP,IT, RamponiFerrantePavon2009,RamponiFerrantePavon2010}).

We restrict our attention to the present setting where $\{y_t\}_{t\in\mZ}$ is scalar as before and so are the filters.  We assume estimates for the output covariances, hence, the values $w_k$'s. Like before, we now denote by $\cF_{\bf z,w}$ the family of power spectra for the process $\{y_t\}_{t\in\mZ}$ which are consistent with these values and we are interested in assessing the size of this family as a measure of our spectral uncertainty.

The following proposition can be derived almost verbatim as Proposition \ref{prp:supKnorm}. See \cite{KarlssonGeorgiou2012} for an independent proof.

\begin{prop} \label{prp:supKnormNP}
Let $z_0,\ldots, z_n$ and $w_0,\ldots, w_n$ be such that the Pick matrix $P$ in \eqref{eq:pick} is positive and let $K\subset \mD$ be closed. Then
\begin{align*}
&\rho_{\delta_K}(\cF_{\bf z,w})=\\
&\max_{z\in K}\left\{2\left(\left|\frac{\frac{1}{1-z\bar{z}}+\langle b_z,d_z\rangle_{P^{-1}}}{\langle b_z,b_z\rangle_{P^{-1}}}\right|^2-\frac{\langle d_z,d_z\rangle_{P^{-1}}}{\langle b_z,b_z\rangle_{P^{-1}}}\right)^{\frac{1}{2}}\right\},
\end{align*}
where
\[
b_z=\left(\begin{array}{l}
\frac{1}{1-z_0\bar z}\\
\frac{1}{1-z_1\bar z}\\
\vdots\\
\frac{1}{1-z_n\bar z}
\end{array}\right),\,
d_z=-\left(\begin{array}{l}
\frac{w_0}{1-z_0\bar z}\\
\frac{w_1}{1-z_1\bar z}\\
\vdots\\
\frac{w_n}{1-z_n\bar z}
\end{array}\right),
\]
and $\langle x,y\rangle_{P^{-1}}$ denote the inner product 
\[
\langle x,y\rangle_{P^{-1}}:=y^* P^{-1} x.
\]  As before,
$\rho_{\delta_K}(\cF_{\bf z,w})$ is attained as the distance
between two elements of $\cF_{\bf z,w}$ which are both singular
with support containing at most $n+1$ points.
\end{prop}

As in the covariance case {\it a priori} bounds on the uncertainty may be calculated.
\begin{thm}\label{thm:THREEbound}
Following the notation of Proposition \ref{prp:supKnormNP}, let $z_0=0$ and
\[
B_{\bf z}(z)=\prod_{k=0}^n \frac{z-z_k}{1-\bar z_k z}.
\]
Then 
\begin{equation}\label{eq:THREEbound}
\rho_{\delta_K}(\cF_{\bf z,w})\le \max_{z\in K}\frac{4w_0|B_{\bf z}(z)|}{1-|z|^2}.
\end{equation}
Further, \eqref{eq:THREEbound} holds with equality if and only if 
\[
w_k=w_0(1+z_k\bar \alpha)/(1-z_k\bar \alpha) \mbox{ for }k=1,\ldots, n
\]
for some $\alpha\in K$ maximizing
$|B_{\bf z}(\alpha)|/(1-|\alpha|^2)$.
\end{thm}
\begin{proof}
The proof is given in the appendix.
\end{proof}
Here the {\it a priori} bound depends on the interpolation points ${\bf z}$, in addition to $K$, the model order $n$, and the total spectral mass $w_0$. Therefore, by minimizing the right hand side of \eqref{eq:THREEbound} with respect to the ${\bf z}$, one can find the filter bank with the smallest {\it a priori} uncertainty in the metric $\delta_K$. This will be exploited in the following example to tune the filter-bank poles.

\section{Uncertainty in the THREE framework with optimal filter selection}\label{sec:THREEex}
From this vantage point we now take up an example as before, with closely spaced sinusoids, and compare two alternative formalisms, one based on Toeplitz covariances and
the other based on generalized statistics.

Consider the stochastic process
\[
y_t=\frac{\cos(0.5t+\varphi_1)+\cos(0.6t+\varphi_2)}{2}+\cos(t+\varphi_3)+w_t+\frac{1}{3}w_{t-1},
\]
with two closely-spaced spectral lines at $0.5$ rad/s and $0.6$ rad/s superimposed with a spectral line in $1$ and colored noise.
We choose as metric $\delta_K$, with $K\subset \mD$ proximal to
the region where high resolution is desired -- i.e., near $0.5$ rad/s where the two closely-spaced sinusoids reside. More specifically, we take\footnote{$\mT$ denotes as before the unit circle.}
\[
K:=\{0.65e^{\pm 0.5i}+0.25\mT\}.
\]
This is depicted by the two circles in Figure \ref{fig:ex2circle}.

We compare the maximum entropy
spectral estimate $d\mu_{\rm ME}$ constructed using the covariances $c_0, c_1,
\ldots, c_{20}$, with the spectral
estimate  $d\mu_{_{\rm THREE}}$ which is based on the output statistics of the  filter bank of $G_k(z)$'s. We select $n=10$ and filter-bank poles that minimize\footnote{The pole $z_0=0$ and total spectral mass $w_0=1$ are assumed fixed. The bound in \eqref{eq:THREEbound} is then minimized over $z_1,\ldots, z_n$. Since RHS of \eqref{eq:THREEbound} is nonconvex in $\bf z$ only local minimum is guaranteed.} the {\it a priori} uncertainty bound \eqref{eq:THREEbound}. The filter poles, indicated by ``$\times$'' in Figure~\ref{fig:ex2circle}, are
 \begin{align*}
z_k \in&\{0,\; 0.581 \pm 0.480i,\; 0.681 \pm 0.470i,   \\& 0.738 \pm 0.422i, \; 0.755 \pm 0.271i, \;0.765 \pm 0.357i \}.
 \end{align*}
The THREE-spectrum is a ``maximum entropy'' distribution which is now consistent with statistics other than the usual autocorrelation ones ($d\mu_{_{\rm THREE}}$ is the so called ``central solution'' of the Nevanlinna-Pick analytic interpolation theory\footnote{Software is available at\\ {\tt http://www.ece.umn.edu/$\sim$georgiou/code/spec\_analysis.tar}} to distributions in $\cF_{\bf z,w}$).

The {\it a priori} bounds on the uncertainty provided by Theorems \ref{thm:COVbound} and \ref{thm:THREEbound} are
\begin{align*}
\rho_{\delta_K}(\cF_{\bf z,w})&\le 0.151w_0= 0.468 \mbox{ and}\\
\rho_{\delta_K}(\cF_{{\bf c}_{0:n}})&\le 2.304\,c_0\,= 7.167,
\end{align*}
respectively. In our example $w_0=c_0=28/9$.
This shows that the {\it a priori} bound  
on the uncertainty set with respect to $\delta_K$ is considerably smaller when the
THREE formalism is applied.

The two spectral estimates together with the true power spectrum are depicted in Figure
\ref{fig:ex2spectra}. It can be seen that the two closely-spaced lines
are not discernible in $d_{\rm ME}$. On the other hand, they are quite clearly distinguishable via THREE. This is due to the choice of the dynamics {\bf z}. As can be seen from the figure, the resolution of $d\mu_{_{\rm THREE}}$ is substantially higher than that of $d\mu_{\rm ME}$ in the vicinity of $0.5$ rad/s.
We would also like to compare the size of the uncertainty set for the two scenarios. The size of the respective diameters are
\[
\rho_{\delta_K}(\cF_{\bf z,w})= 0.194 \mbox{ and }
\rho_{\delta_K}(\cF_{{\bf c}_{0:n}})= 2.831.
\]
Thus, when measured using $\delta_K$, the uncertainty set using the THREE formalism
 is considerably smaller.
Figure \ref{fig:ex2bounds} displays the Poisson integral of the true power spectrum evaluate on $K$, and the corresponding bounds\footnote{Since $K$ is formed out of two circles symmetrically located with respect to the real axis of the complex plane, plots are identical for the two components of $K$.}.

\begin{figure}
\centering
\includegraphics[scale=0.45]{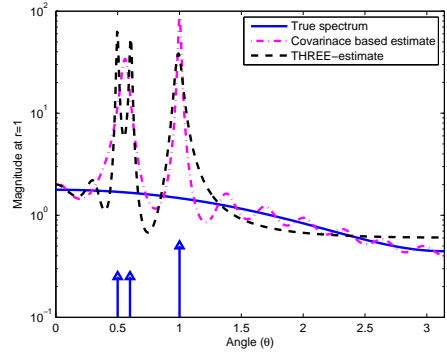}
\caption{{True spectrum $d\nu$ (solid) and estimated spectra $d\mu_{_{\rm THREE}}$ (dashed-dotted) and $d\mu_{\rm ME}$ (dashed).}}
\label{fig:ex2spectra}
\end{figure}

\begin{figure}
\centering
\includegraphics[scale=0.45]{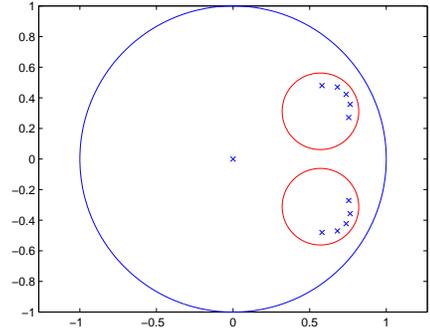}
\caption{{Set $K$ (solid red) and points $z_k$ ($\times$ in blue).}}
\label{fig:ex2circle}
\end{figure}

\begin{figure}
\centering
\includegraphics[scale=0.45]{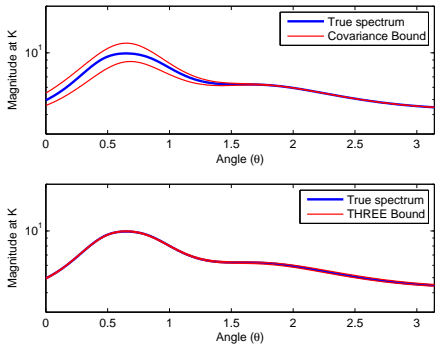}
\caption{{Bounds on estimates on $K$ based on covariances (top) and the THREE formalism (bottom), respectively. 
}}
\label{fig:ex2bounds}
\end{figure}

\section{Conclusions and future directions}\label{sec:conclutions}
The choice of a metric is key to any quantitative scientific theory.
Identification of power spectra is often based on second-order statistics (moments), and therefore, it is natural to metrize the space of power spectra
in a way that respects continuity of moments.
There is a variety of such weakly continuous metrics ---metrics which localize ``spectral
mass''. We presented various choices and focused on a particular metric, $\delta_K$, which is
amenable to quantifying the size of the uncertainty set.
We envision that this, and similar metrics, can be used as tools for assessing uncertainty and robustness in modeling and spectral analysis.
We further expect that the theory will be of use in filter design and in quantifying the notion of resolution--as this is naturally connected to the size of the spectral uncertainty set. Finally, we expect that these metrics will conform with other subjective measures rooted in perceptual qualities of signals (cf.\ \cite[Example 10]{GKT}).

Interest in weak continuity is not new.
Indeed, a classical weakly continuous metric is the L\'{e}vy-Prokhorov metric \cite{Prokhorov} and it is well known that the periodogram
converges weakly as the sample size goes to infinity (see, e.g., \cite{Parzen}). 
Yet, appropriate weakly continuous metrics that can be used to quantify uncertainty have not received much attention ---the commonly used ``total variation,'' Itakura-Saito, and other distance measures are not weakly continuous.
Besides the relevance in uncertainty quantification and in filter design (cf.\ Section \ref{sec:THREEex}), 
computationally amenable and easy-to-use metrics may provide a useful geometric setting for
modeling slowly time-varying processes and for integrating data from disparate sources (see, e.g., \cite{Jiang,Jiang1,Jiang2,rudoy,rudoy2}).

\section*{Appendix}
\begin{proof}{\em [Theorem \ref{thm:weakstardistance}]}\\
The canonical neighborhood basis for a point $d\nu$ in the weak topology on $\fM$ consists of sets of the type
\begin{eqnarray*}
&&N(d\nu,\{g_k\}_{k=1}^n,\epsilon)\\
&&\hspace{20pt}=\Big\{d\mu\ge 0\,:\, \left|\int_\mT g_k(d\nu-d\mu)\right|<\epsilon, k=0,1,\dots,n\Big\},
\end{eqnarray*}
where $g_k$ are continuous functions on $\mT$ for $k=0,\dots,n$. 
To establish the theorem we prove that the neighbourhood basis
\[
\fN(d\nu)=\left\{N(d\nu,\{g_k\}_{k=0}^n,\epsilon)\,:\, \epsilon>0, n\in \mN, \{g_k\}_{k=0}^n\subset C(\mT)\right\}
\]
is equivalent to the basis 
\[
\fF(d\nu)=\left\{\cF_{\bczn,\epsilon} \,:\, \epsilon>0, n\in \mN,   c_k=\int_\mT z^{-k} d\nu, k=0,\dots,n\right\}.
\]
First note that $\fN(d\nu) \supset\fF(d\nu)$, and hence the weak topology is at least as strong as the topology induced by $\fF(d\nu)$.
To establish the other direction, let
$N$ be an arbitrary set in $\fN(d\nu)$. To show the equivalence, it is enough to show that 
there exists $n\in \mN$ such that 
$\cF_{\bc_{0:n},n^{-1}} \subset N.$

Let $\delta$ be a weakly continuous metric
for $\fM$ and choose $\epsilon$ so that the $B_\delta(d\nu, \epsilon)=\{d\mu\ge 0\,:\, \delta(d\mu, d\nu)<\epsilon\}\subset N$. Next, take $d\mu_\ell\in \cF_{\bc_{0:\ell},\ell^{-1}}$ with 
\begin{equation}\label{eq:boundmun}
\delta(d\mu_\ell, d\nu)\ge \frac{1}{2}\sup \{\delta(d\mu, d\nu)\,:\, d\mu\in\cF_{\bc_{0:\ell},\ell^{-1}}\}
\end{equation}
for  $\ell\ge 1$.
Since $\cF_{\bc_{0:\ell},\ell^{-1}}\subset \{d\mu\,:\, \mu(\mT)\le\nu(\mT)+1)$, which is weakly compact, there is a convergent subsequence of $d\mu_k$ that converges to $d\hat\mu$ (by Banach-Alaoglu \cite{Rudin1971}).
Note that $\cF_{\bc_{0:\ell},\ell^{-1}}\supset {\rm closure} (\cF_{\bc_{0:\ell+1},(\ell+1)^{-1}})\supset \cF_{\bc_{0:\ell+1},(\ell+1)^{-1}}$, hence $d\hat\mu\in\cF_{\bc_{0:\ell},\ell^{-1}}$ for any $\ell$. It then follows that $d\hat\mu=d\nu$ since the trigonometric moment problem is determinate (by Riesz-Herglotz, see \cite{Akhiezer}). Let $n$ be such that $\delta(d\mu_n,d\nu)<\epsilon/2$, then by \eqref{eq:boundmun} we have that $\cF_{\bc_{0:n},n^{-1}}\subset B_\delta(d\nu, \epsilon)\subset N$.
We have thus shown that the topology induced by the neighbourhood basis $\fF(d\nu)$
is the weak topology, and hence $\delta$ is weakly continuous if and only if \eqref{eq:weaklimit} holds.
\end{proof}

\begin{proof}{\em [Proposition \ref{prp:locmetric}]}\\
It is clear that condition (a) holds if and only if $\delta(d\mu_0,d\mu_1)$ is positive whenever $d\mu_0\neq d\mu_1$. The triangle inequality and symmetry always holds for such $\delta$, so we only need to show that condition (b) holds if and only if $\delta$ is weakly continuous. 

We will show that condition (b) implies that $\delta$ is weakly continuous by contradiction. Assume therefore that condition (b) holds, but that $\delta$ is not weakly continuous. 
Then there exists $d\mu_k\to d\mu$ weakly such that $\delta(d\mu_k, d\mu)>\epsilon$, $k=1,2,\dots$, and hence there exists $g_{\xi_k},{\xi_k}\in K$, such that
\[
\epsilon<\left|\int_\mT g_{\xi_k}(d\mu_k-d\mu)\right|,\; k=1,2,\dots.
\]
To this end we use the Arzel\`a-Ascoli theorem (see e.g., \cite[page 102]{kolmogorov}) which states that a set of functions is relatively compact\footnote{A set is relatively compact if its closure is compact.} in $C(\mT)$
if and only if the set of functions is uniformly bounded and equicontinuous.
Therefore, since (b) holds, the set $\{g_\xi\}_{\xi\in K}$ is relatively compact in $C(\mT)$, and there is a subsequence $(g_\ell,d\mu_\ell)$ 
of $(g_{\xi_k},d\mu_k)$ such that $g_\ell\to g\in C(\mT)$. A contradiction follows, since
\begin{eqnarray*}
\epsilon&<&\left|\int_\mT g_\ell(d\mu_\ell-d\mu)\right|\\
&\le&\|g_\ell-g\|_\infty\int_\mT|d\mu_\ell-d\mu|+\left|\int_\mT g(d\mu_\ell-d\mu)\right|\\
&\to& 0 \mbox{ as } \ell\to \infty,
\end{eqnarray*}
and hence $\delta$ is weakly continuous whenever condition (b) holds.

Next, we show that (b) holds if $\delta$ is weakly continuous, and once again we use contradiction. That is, we show that 
 if (b) fails to be true then $\delta$ is not weakly continuous. If $\{g_\xi\}_{\xi\in K}$ is not equicontinuous, 
then there exists an $\epsilon>0$ such that for any $k=1,2,\ldots$ one can find $\theta_k,\phi_k\in \mT$, and $\xi_k\in K$, that satisfies
\begin{equation}\label{eq:equicont}
|\theta_k-\phi_k|<\frac{1}{k} \mbox{ and } |g_{\xi_k}(\theta_k)-g_{\xi_k}(\phi_k)|>\epsilon.
\end{equation}
Let $(\theta_\ell,\phi_\ell)$ be a subsequence of $(\theta_k,\phi_k)$ such that $\theta_\ell\to\theta_0\in\mT$ as $\ell \to \infty$, and let $d\mu_\ell$ and $d\nu_\ell$ be the measures that consist of a unit mass in $\theta_\ell$ and $\phi_\ell$, respectively. From \eqref{eq:equicont}  it follows that $\phi_\ell\to\theta_0$, and hence that $d\mu_\ell\to d\mu_0$ and $d\nu_\ell\to d\mu_0$ weakly, where $d\mu_0$ is the measure that consists of a unit mass in $\theta_0$. From \eqref{eq:equicont} it follows that
\begin{eqnarray*}
\delta(d\mu_\ell, d\mu_0)+\delta(d\nu_\ell, d\mu_0)&\ge&\delta(d\mu_\ell, d\nu_\ell)\\
&\ge&|g_{\xi_\ell}(\theta_\ell)-g_{\xi_\ell}(\phi_\ell)|>\epsilon.
\end{eqnarray*}
From this, it is evident that $\delta$ is not weakly continuous since  both $\delta(d\mu_\ell, d\mu_0)$ 
and $\delta(d\nu_\ell, d\mu_0)$ cannot converge to $0$.

Similarly, if $\{g_\xi\}_{\xi\in K}$ is not uniformly bounded, 
then for any $k=1,2,\ldots$ one can find $\theta_k\in \mT$ and $\xi_k\in K$ such that
\begin{equation}\label{eq:ubound}
|g_{\xi_k}(\theta_k)|>k.
\end{equation}
Let $d\mu_k$ be the measures that consist of a unit mass in $\theta_k$. Therefore, the metric $\delta$ is not weakly continuous since $\frac{1}{k}d\mu_k \to 0$ weakly, while $\delta(\frac{1}{k}d\mu_k, 0)>1$ for all $k$.
\end{proof}

\begin{proof}{\em [Proposition \ref{prp:eqlimits}]}\\
 $(a) \Rightarrow (b)$
$\mu_k\rightarrow \mu$ weakly is equivalent to $\int_{-\pi}^\pi f(t)d\mu_k(t)\rightarrow \int_{-\pi}^\pi f(t)d\mu(t)$ for all periodic continuous functions $f(t)$. For all $z=re^{i\theta}\in\mD$, $P_r(\theta-t)$ is periodic and continuous, hence
\begin{eqnarray*}
u_k(z)&=&\frac{1}{2\pi}\int_{-\pi}^\pi P_r(\theta-t)d\mu_k(t)\\ 
&\rightarrow& \frac{1}{2\pi}\int_{-\pi}^\pi P_r(\theta-t)d\mu(t)=u(z).
\end{eqnarray*}

$(b) \Rightarrow (c)$. 
For $r<1$, $|u_k(re^{i\theta})|\le \frac{1+r}{1-r}|\mu_k|(\mT)$. 
Since $u_k(re^{i\theta})\to u(re^{i\theta})$ pointwise for all $\theta$, it follows from bounded convergence that $\int_{\pi}^\pi|u_k(re^{i\theta})- u(re^{i\theta})|d\theta\to 0$. 
Further more, $\forall k,r$, $\int_{\pi}^\pi|u_k(re^{i\theta})- u(re^{i\theta})|d\theta \le 2\pi(|\mu_k|(\mT)+|\mu|(\mT))$ which is uniformly bounded, hence
\[
\int_0^1\int_{-\pi}^\pi |u_k(re^{i\theta})- u(re^{i\theta})|d\theta rdr\to 0
\]
by dominated convergence.

$(c) \Rightarrow (d)$. Let $K \subset \mD$ be a compact set. Then there exist an $\epsilon > 0$ such that $B_\epsilon(z_0)=\{z:|z-z_0|<\epsilon\}\subset \mD$ for all $z_0\in K$. Now by the mean value property of harmonic functions we have 
\begin{eqnarray*}
u(z_0)&=&\frac{2}{\epsilon^2}\int_0^\epsilon u(z_0) rdr\\
&=&\frac{2}{\epsilon^2}\int_0^\epsilon \frac{1}{2\pi}\int_{\pi}^\pi u(z_0+r e^{i\theta})d\theta rdr\\
&=&\frac{1}{\pi\epsilon^2}\int_{B_\epsilon(z_0)} u(z) dxdy.
\end{eqnarray*}
Of course the same equality holds for $u_k(z_0)$
\[
u_k(z_0)=\frac{1}{\pi\epsilon^2}\int_{B_\epsilon(z_0)} u_k(z) dxdy.
\]
For any $z_0\in K$ the difference between the harmonic functions is bounded by
\begin{align*}
|u_k(z_0)-u(z_0)|&\le\frac{1}{\pi\epsilon^2}\int_{B_\epsilon(z_0)} |u_k(z)-u(z) |dxdy\\
&\le\frac{1}{\pi\epsilon^2}\int_\mD |u_k(z)-u(z) |dxdy.
\end{align*}
By $(c)$ the difference goes to zero uniformly in $K$.

$(d) \Rightarrow (a)$.
Let $f\in C(\mT)$. For any bounded measure $\nu\in\cF$ and corresponding harmonic function $v(z)=P[\nu](z)$ Fubini's theorem gives
\begin{eqnarray}
\int_{-\pi}^\pi f(t) v(re^{it})dt & = & \int_{-\pi}^\pi\frac{1}{2\pi}\int_{-\pi}^\pi P_r(\theta-t)f(t)dt d\nu(\theta) \nonumber \\ & = & \int_{-\pi}^\pi P[f(t)dt](re^{i\theta})d\nu(\theta).\nonumber
\end{eqnarray}
Since $f$ is periodic and continuous, $P[f(t)dt](re^{i\theta})$ converges uniformly to $f(\theta)$, hence 
\begin{align*}\label{eq:uniformlyweak}
&\left|\int_{-\pi}^\pi f(t) v(re^{it})dt - \int_{-\pi}^\pi f(t) d\nu(t)\right|\le\\ 
&\hspace{2.5cm}\|P[f](re^{it})-f(t)\|_\infty|\nu|(\mT)
\end{align*}
 converges to zero independent of the measure $\nu$. This shows that for an arbitrary $\epsilon>0$ there exists an $0<r<1$ such that 
\[
\left|\int_{-\pi}^\pi f(t) v(re^{it})dt - \int_{-\pi}^\pi f(t) d\nu(t)\right|<\frac{\epsilon}{3}
\]
for $\nu\in\{\mu,\mu_1,\mu_2,\ldots\}$. Further more, since $u_k\rightarrow u$ uniformly on $\{z:|z|\le r\}$, it is possible to find an $k_{r,\epsilon}$ be such that 
\[
\left|\int_{-\pi}^\pi f(t) u_k(re^{it})dt - \int_{-\pi}^\pi f(t) u(re^{it})dt\right|<\frac{\epsilon}{3}
\]
 for all $k>k_{r,\epsilon}$.
By the triangle inequality we have 
\[
\left|\int_{-\pi}^\pi f(t)d\mu_k(t)- \int_{-\pi}^\pi f(t)d\mu(t)\right| < \epsilon
\]
for all $k>k_{r,\epsilon}$. Since $\epsilon$ was chosen arbitrarily, $\left|\int_{-\pi}^\pi f(t)d\mu_k(t) - \int_{-\pi}^\pi f(t)d\mu(t)\right|\rightarrow 0$ as $k\rightarrow \infty$, and weak convergence follows.
\end{proof}

\begin{proof}{\em [Proposition \ref{prp:supKnorm}]}\\
There exists an analytic function $f(z)=H[d\mu](z), d\mu\in \cF_{{\bf c}_{0:n}}$, such that 
$f(z)=w_z$ if and only if its associated Pick matrix is nonnegative \cite{KP}, i.e.
\begin{equation}\label{eq:pickmatrix}
\left(\begin{array}{cc}
2T_n & b_zw_z -d_z\\
\bar{w}_zb^*_z-d^*_z& \frac{w_z+\bar{w}_z}{1-z\bar{z}}\end{array}\right)\ge 0.
\end{equation}
By using Schur's lemma and completing the square, we arrive at
\begin{align} 
&\left|w_z-\frac{\frac{2}{1-z\bar{z}}+\langle d_z,b_z\rangle_{T^{-1}}}{\langle b_z,b_z\rangle_{T^{-1}}}\right|^2 \nonumber\\
& \le\left|\frac{\frac{2}{1-z\bar{z}}+\langle b_z,d_z\rangle_{T^{-1}}}{\langle b_z,b_z\rangle_{T^{-1}}}\right|^2-\frac{\langle d_z,d_z\rangle_{T^{-1}}}{\langle b_z,b_z\rangle_{T^{-1}}},\label{eq:discw} 
\end{align}
where equality holds if and only if the Pick matrix (\ref{eq:pickmatrix}) is singular. From this, the first part of Proposition \ref{prp:supKnorm} follows. Since the maximum is obtained when equality holds in (\ref{eq:discw}), the associated Pick matrices are singular. Hence the solutions are unique and correspond to measures with support on $n+1$ points \cite[Proposition 2]{Geor00}.
\end{proof}

\subsection*{Background on orthogonal polynomials and Schur coefficients}
Let $\bc$ be a nonnegative covariance sequence with corresponding measure $d\mu$ and consider the inner product
\[
\langle a(z),b(z)\rangle = \frac{1}{2\pi}\int_\mT a(e^{i\theta})\overline{b(e^{i\theta})}d\mu(\theta).
\]
The so-called {\em orthogonal polynomials (of the first kind)} $\phi_k(z)$ \cite{Ger}
are (uniquely defined) monic polynomials with $\deg \phi_k(z)=k$, $k=0,1\ldots$, which are 
orthogonal with respect to $\langle \cdot,\cdot\rangle$. They are shown \cite{Ger} to satisfy the recursion 
\begin{equation} \label{eq:rec1}
\begin{array}{ccc}
\phi_{k+1}(z) &=& z\phi_k(z)-\bar{\gamma}_k\phi_k(z)^*,\\
\phi_{k+1}(z)^* &=& \phi_k(z)^*-z\gamma_k\phi_k(z), 
\end{array}
\end{equation}
where $\phi_k(z)^*=z^k\overline{\phi_k(\bar{z}^{-1})}$ and $\{\gamma_k\}_{k=1}^\infty$ are the so-called {\em Schur parameters}.

The {\em orthogonal polynomials of the second kind} are defined by 
\[
\psi_k(z)=\frac{1}{c_0}[(\overline{f(\bar{z}^{-1})})\phi_k(z)]_+,
\] 
where $[\cdot]_+$ denote ``the polynomial part of''. They are also ``orthogonal polynomials'' but with respect to a certain ``inverted'' covariance (corresponding to the negative of the original Schur parameters, cf.\ \cite{Ger}) and satisfy the recursion
\begin{equation}\label{eq:rec2}
\begin{array}{ccc}
\psi_{k+1}(z) &=& z\psi_k(z)+\bar{\gamma}_k\psi_k(z)^*,\\
\psi_{k+1}(z)^* &=& \psi_k(z)^*+z\gamma_k\psi_k(z).
\end{array}
\end{equation}

The positive-real function $f(z)=H[d\mu](z)$ may be expressed using the orthogonal polynomials as 
\begin{equation}\label{eq:fparam}
f(z)=c_0\frac{\psi_k(z)^*+zs_{k+1}(z)\psi_k(z)}{\phi_k(z)^*-zs_{k+1}(z)\phi_k(z)},
\end{equation}
where $s_{k+1}(z)$ belong to the Schur class ${\cal S}$, i.e. the class of analytic functions on $\mD$ uniformly bounded by $1$. 
Equations (\ref{eq:rec1}-\ref{eq:rec2}) lead to
\begin{eqnarray}\label{eq:shurrec}
f(z)=c_0\frac{1+zs_1(z)}{1-zs_1(z)},\quad 
s_{k}(z)=\frac{\gamma_{k}+zs_{k+1}(z)}{1+z\bar{\gamma}_{k}s_{k+1}(z)},
\end{eqnarray}
for $k=1,2,\ldots$.
For a complete exposition on orthogonal polynomials and Schur's algorithm see \cite{Akhiezer, Ger, GrenanderSzego}.

\begin{proof}{\em [Theorems \ref{thm:COVbound} and \ref{thm:THREEbound}]}

Following our earlier notation, let 
\[\rho_{\delta_\alpha}(\cF_\bczn)=\max\{|P[d\mu_0](\alpha)-P[d\mu_1](\alpha)|: d\mu_0, d\mu_1\in \cF_\bczn\}
\]
be the uncertainty diameter at the point $\alpha\in K$. 
By using \eqref{eq:fparam} and noting that 
\[
P[d\mu](\alpha)={\mathfrak Re} H[d\mu](\alpha)={\mathfrak Re} f(\alpha),
\]
the diameter $\rho_{\delta_\alpha}(\cF_\bczn)$ is equal to the diameter of the disc 
\begin{equation}\label{eq:circlr}
c_0 \frac{\psi_n(\alpha)^*+
\alpha s_{n+1}(\alpha)\psi_n(\alpha)}{\phi_n(\alpha)^* -\alpha s_{n+1}(\alpha)
\phi_n(\alpha)} \, : \, s_{n+1}(\alpha)\in \overline\mD,
\end{equation}
where $\phi_n$ and $\psi_n$ are specified via \eqref{eq:rec1}, \eqref{eq:rec2},  by the Schur sequence $(\gamma_1, \ldots, \gamma_n)$ corresponding to $\bczn$ (see \cite{GrenanderSzego}). Denote by  $r_n$ the radius of \eqref{eq:circlr}, and hence $\rho_{\delta_\alpha}(\cF_\bczn)=2r_n$.

For $n=0$, we have $\phi_0(z)=\psi_0(z)= 1$. The expression in \eqref{eq:circlr} is
\[
c_0\frac{1+\alpha s_1}{1-\alpha s_1}=c_0\frac{1+|\alpha |^2}{1-|\alpha |^2}+\frac{2c_0\alpha}{1-|\alpha|^2}\frac{-\bar \alpha +s_1}{1-\alpha s_1},
\]
hence $r_0=2c_0\alpha /(1-|\alpha|^2)$, where $s_k$ without argument denotes $s_k(\alpha)$.
Next, consider the radius of \eqref{eq:circlr} for $n={k-1}$. The set \eqref{eq:circlr} is the range of a M\"obius transform applied to $s_k\in\overline \mD$, and may be represented as
\begin{equation}\label{eq:Mrep}
M_{k-1}+e^{i\theta_{k-1}}r_{k-1}\frac{v_{k-1}+s_k}{1+\bar v_{k-1}s_k}
\end{equation} 
where $M_{k-1}$ and $r_{k-1}$ are the center and radius of the disc, respectively, and where $\theta_{k-1}\in(-\pi, \pi]$,  $v_{k-1}\in\mD$.
From the recursion \eqref{eq:shurrec}, is can be seen that
\[
\frac{v_{k-1}+s_k}{1+\bar v_{k-1}s_k}=\frac{1+\bar\gamma_k v_{k-1}}{1+\gamma_k \bar v_{k-1}}\;\frac{\eta_k+\alpha s_{k+1}}{1+\bar \eta_k\alpha s_{k+1}}
,\]
where 
\[
\eta_k=\frac{v_{k-1}+\gamma_{k}}{1+\bar\gamma_{k}v_{k-1}}.
\]
The set \eqref{eq:circlr} for $n=k$ is therefore
\[
M_{k-1}+e^{i\theta_{k-1}} \frac{1+\bar\gamma_k v_{k-1}}{1+\gamma_k \bar v_{k-1}} r_{k-1} \frac{\eta_k+\alpha s_{k+1}}{1+\bar \eta_k\alpha s_{k+1}}\,: \; s_{k+1}\in \mD.
\]
A M\"obius transformation $(a+bs)/(c+ds)$, with $|c|>|d|$, maps the unit disc to a disc of radius $|ac-bd|/(|c^2|-|d|^2)$. 
Therefore, the radius 
\[
r_k=r_{k-1}|\alpha|\frac{1-|\eta_k|^2}{1-|\eta_k|^2|\alpha|^2}=r_{k-1}|\alpha |\left(1-\frac{|\eta_k|^2(1-|\alpha|^2)}{1-|\eta_k|^2|\alpha|^2}\right)
\]
is maximized when $\eta_k=0$, or equivalently when $\gamma_k=-v_{k-1}$. Hence, $r_k\le r_{k-1}|\alpha|$ with equality if $\gamma_k=-v_{k-1}$.
By induction, the maximal radius is given by 
\[r_n=r_0|\alpha|^n=2c_0|\alpha|^{n+1}/(1-|\alpha|^2).
\] Furthermore, $\gamma_k=-v_{k-1}$ in the recursion \eqref{eq:shurrec} correspond to the Schur parameters  
$\gamma_1=\bar \alpha,$ and $\gamma_k=0 \mbox{ for } k=2,\ldots, n.$ 
This leads to the covariance sequence 
$\bczn=(c_0, c_0 \bar \alpha,\ldots, c_0\bar \alpha^n )$.
Since $\delta_K$ is defined as the maximal diameter over all $\alpha\in K$, the inequality
\[
\rho_{\delta_K}(\cF_\bczn)=\max_{\alpha\in K}\rho_{\delta_\alpha}(\cF_\bczn)\le\max_{\alpha\in K}\frac{4|\alpha|^{n+1}c_0}{1-|\alpha|^2}
\]
holds and is achieved for $\bczn=(c_0, c_0 \bar \alpha,\ldots, c_0\bar \alpha^n )$ where $\alpha\in K$  maximizes $|\alpha|^{n+1}/(1-|\alpha|^2)$.

In the Nevalinna-Pick case, the recursion is identical to \eqref{eq:shurrec} except that $z$ is replaced by the inner factor $\xi_k(z)=(z_k-z)/(1-\bar z_k z)
$
\[
f(z)=w_0\frac{1+zs_1(z)}{1-zs_1(z)},\quad 
s_{k}(z)=\frac{\gamma_{k}+\xi_k(z) s_{k+1}(z)}{1+\bar{\gamma}_{k}\xi_k(z)s_{k+1}(z)},
\]
for $k=1,2,\ldots, n$  \cite{Garnett}. The argument here is analogous to the covariance case. The shrinkage of the radius is $r_k\le r_{k-1}|\xi_k(\alpha)|$, with equality when the parameters in the recursion are $(\bar \alpha, 0, \ldots, 0)$, as in the covariance case. The bound of the uncertainty diameter at $\alpha$ then becomes 
\begin{equation}\label{eq:THREEboundAlpha}
2r_n\le\frac{4w_0|\alpha|\prod_{z=1}^n |\xi_k(\alpha)|}{1-|\alpha|^2}=\frac{4w_0|B_{\bf z}(\alpha)|}{1-|\alpha|^2}
\end{equation}
which is attained when $w_k=w_0(1+z_k\bar \alpha)/(1-z_k\bar \alpha)$ for $k=1,\ldots, n$. Since $\rho_{\delta_\alpha}(\cF_{\bf z,w})=2r_n$, maximizing \eqref{eq:THREEboundAlpha} for $\alpha\in K$ gives the bound \eqref{eq:THREEbound}. 
\end{proof}

\end{document}